\documentclass{article}
\usepackage{amssymb,amsmath,amsfonts, amsthm}    %ams
\usepackage[english]{babel}
\usepackage[utf8x]{inputenc}
\usepackage{graphicx}
\usepackage{pstricks}
\usepackage{multido}
\usepackage{verbatim}
\usepackage[ruled,vlined]{algorithm2e}
\usepackage{enumerate}
\usepackage{dsfont}
\usepackage{tikz}
\usepackage{nomencl}
\usepackage{complexity}
\makenomenclature
\newcommand{\Proba}{{\mathbb P}}

\newtheorem{theorem}{Theorem}
\newtheorem{corollary}{Corollary}

\newtheorem{lemma}{Lemma}
\newtheorem{definition}{Definition}

{\everymath{\displaystyle\everymath{}}\array}%
{\endarray}

\def\hyann#1{}

\def\hdavid#1{}

\def\hpierre#1{}
\newcommand{\pro}[1]{{\mathbb P}\left(#1\right)}

\newcommand{\val}{\mbox{Val}}

\title{Finding Optimal Strategies of Almost Acyclic Simple Stochatic Games}
\author{  David Auger, Pierre Coucheney, Yann Strozecki }

\date{}

\begin{document}

\maketitle

\begin{center}
 
    {\tt \{david.auger, pierre.coucheney, yann.strozecki\}@uvsq.fr}

    \vskip .3cm

    PRiSM,  Université de Versailles Saint-Quentin-en-Yvelines
     
     Versailles, France
    
\end{center}

\begin{abstract}
    The optimal value computation for turned-based stochastic games with reachability objectives, also known as \emph{simple stochastic games}, is one of the few problems in $\NP \cap \coNP$ which are not known to be in $\P$. However,
 there are some cases where these games can be easily solved, as for instance when the underlying graph is acyclic. 
 In this work, we try to extend this tractability to several classes of 
 games that can be thought as "almost" acyclic. We give some fixed-parameter tractable or polynomial algorithms 
 in terms of different parameters % measuring how far a graph is from being acyclic,
 such as the number of cycles or the size of the minimal feedback vertex set.
 \end{abstract}

%%%%%%%%%%%%%%%%%%%%%%%%%%%%%%%%%%%%%%%%%%%%%%%%%%%%%%%%%%%%%%%%%%%%%%%%%%%%%%%%%%%%%%%%%%%%%version llncs%%%%%%%%
 %\keywords{algorithmic game theory $\cdot$ stochastic games $\cdot$ FPT algorithms}
%%%%%%%%%%%%%%%%%%%%%%%%%%%%%%%%%%%%%%%%%%%%%%%%%%%%%%%%%%%%%%%%%%%%%%%%%%%%%%%%%%%%%%%%%%%%%version llncs%%%%%%%%

\begin{center}
   \bf   keywords: algorithmic game theory $\cdot$ stochastic games $\cdot$ FPT algorithms
\end{center}

\section*{Introduction}

A \emph{simple stochastic game}, SSG for short, is a zero-sum, two-player, turn-based version, of the more general 
\emph{stochastic games} introduced by Shapley \cite{shapley1953stochastic}. SSGs
were introduced by Condon \cite{condon1992complexity} and they provide a simple framework
that allows to study the algorithmic complexity issues underlying reachability objectives. A SSG is played by moving a pebble on a graph. Some vertices are divided between players MIN and MAX:
if the pebble attains a vertex controlled by a player then he has to move the pebble along
an arc leading to another vertex. Some other vertices are ruled by chance; typically
they have two outgoing arcs and a fair coin is tossed to decide where
the pebble will go. Finally, there is a special vertex named the $1$-sink, such that if the pebble reaches it player MAX wins, otherwise player MIN wins. 

Player MAX's objective is, given a starting vertex for the pebble, to maximize the probability of winning
against any strategy of MIN.
One can show that it is enough to consider stationary deterministic strategies for both players \cite{condon1992complexity}. % which are finite in number.
Though seemingly simple since the number of stationary deterministic strategies is finite, the task of finding the pair of optimal strategies, or equivalently, of computing the so-called \emph{optimal values} of vertices, is not known to be in $\P$. 

SSGs are closely related to other games such as parity games or discounted payoff games to cite a few~\cite{andersson2009complexity}. Interestingly, those games provide natural applications in model checking of the modal $\mu$-calculus~\cite{stirling1999bisimulation} or in economics. While it is known that they can be reduced to simple stochastic games~\cite{DBLP:journals/corr/abs-1106-1232}, hence seemingly easier to solve, so far no polynomial algorithm are known for these games either. 

Nevertheless, there are some very simple restrictions for SSGs for which the problem of finding optimal strategies is tractable. 
Firstly, if there is only one player, the game is reduced to a Markov Decision Process (MDP) which can be solved by linear programming. 
In the same vein, if there is no randomness, the game can be solved in almost linear time~\cite{andersson2008deterministic}.

As an extension of that fact, there is a Fixed Parameter Tractable (FPT) algorithm, where the parameter is the number of average vertices~\cite{gimbert2008simple}. The idea is to get rid of the average vertices by sorting them according to a guessed order.
Finally, when the (graph underlying the) game is a directed acyclic graph (DAG), the values can be found in linear time by computing them backwardly from sinks.

Without the previous restrictions, algorithms running in exponential time are known. Among them, the Hoffman-Karp~\cite{hoffman1966nonterminating} algorithm proceeds by successively playing a local best-response named switch for one player, and then a global best-response for the other player. Generalizations of this algorithm have been proposed and, though efficient in practice, they fail to run in polynomial time on a well designed example~\cite{friedmann2009exponential}, even in the simple case of MDPs~\cite{fearnley2010exponential}. 
These variations mainly concern the choice of vertices to switch at each turn of the algorithm which is quite similar to the choice of pivoting in the simplex algorithm for linear programming. This is not so surprising since computing the values of an SSG can be seen as a generalization of solving a linear program.
%both problems belong to the class of lp-type problems~\cite{halman2007simple} \mpierre{Je ne sais pas si on en parle car ca ne nous avait pas paru transcendent}. 
The best algorithm so far is a randomized sub-exponential algorithm~\cite{ludwig1995subexponential} that is based on an adaptation of a pivoting rule used for the simplex. 
\subsection*{Our contribution}

In this article, we present several graph parameters such that, 
when the parameter is fixed, there is a polynomial time algorithm
to solve the SSG value problem. More precisely, the parameters we look at will quantify 
how close to a DAG is the underlying graph of the SSG, a case that is solvable in linear time. %Since the DAG case is solvable in linear time, we obtain .
The most natural parameters that quantify the distance to a DAG would be one of the directed versions of the tree-width such as the DAG-width. Unfortunately, we are not yet able to prove a result even for SSG of bounded pathwidth.
In fact, in the simpler case of parity games the best algorithms for DAG-width and clique-width are polynomials
but not even FPT \cite{obdrvzalek2007clique,berwanger2012dag}.
Thus we focus on restrictions on the number of cycles and the size of a minimal feedback vertex set.

First, we introduce in Section~\ref{sec:max-acyclic}  a new class of games, namely  \emph{MAX-acyclic games}, which contains and generalizes the class of acyclic games. We show that the standard Hoffman-Karp algorithm, also known as strategy iteration algorithm, terminates in a linear number of steps for games in this class, yielding a polynomial algorithm to compute optimal values and strategies. It is known that, in the general case, this algorithm needs an exponential number of steps to compute optimal strategies, even in the simple case of Markov Decision Processes \cite{fearnley2010exponential,friedmann2009exponential}. 

Then, we extend in Section~\ref{sec:few_cycles} this result to games with very few cycles, by giving an FPT-algorithm where the parameter is the number of fork vertices which bounds the number of cycles. To obtain a linear dependance in the total number of vertices, we have to reduce our problem to several instances of acyclic games since we cannot even rely on computing the values in a general game. 

% In Section~\ref{sec:few_cycles} we give a linear time FPT algorithm for the value problem on SSGs where the parameter is the number of cycles (Theorem~\ref{th:few-cycles}). To obtain linear time, we have to reduce our problem to several instances of acyclic games, since we cannot even rely on computing the values in a general game.

Finally, in Section~\ref{sec:feedback}, we provide an original method to ``eliminate'' vertices in an SSG. We apply it to obtain a polynomial time algorithm for the value problem on SSGs with a feedback vertex set of bounded size (Theorem~\ref{th:feedback}).

\section{Definitions and standard results}

Simple stochastic games are turn-based stochastic games with reachability objectives involving two players named MAX and MIN.
In the original version of Condon \cite{condon1992complexity}, all vertices except sinks have outdegree exactly two, and there are only two sinks, one with value $0$ and another with value $1$. Here, we allow more than two sinks with general rational values, and more than an outdegree two for positional vertices.

\begin{definition}[SSG]
    A simple stochastic game (SSG) is defined by a directed graph $G=(V,A)$, together with a partition of the vertex set $V$ in four parts $V_{MAX}$, $V_{MIN}$, $V_{AVE}$ and $V_{SINK}$. To every $x \in V_{SINK}$ corresponds a value $\val(x)$ which is a rational number in $\left[0,1\right]$. Moreover, vertices of $V_{AVE}$ have outdegree exactly $2$, while sink vertices have outdegree $1$ consisting of a single loop on themselves.
\end{definition}

In the article, we denote by $n_{M}, n_{m}$ and $ n_{a}$ the size of $V_{MAX}$, $V_{MIN}$ and $V_{AVE}$ respectively and by $n$ the size of $V$.
The set of \emph{positional vertices}, denoted $V_{POS}$, is $V_{POS}=V_{MAX} \cup V_{MIN}$. We now define strategies which we restrict to be stationary and pure, which turns out to be sufficient for optimality. 
Such strategies specify for each vertex of a player the choice of a neighbour. 

\begin{definition}[Strategy]
    A \emph{strategy} for player MAX is a map $\sigma$ from $V_{MAX}$ to $V$ such that
    \[ \forall x \in V_{MAX}, \quad (x,\sigma(x)) \in A.  \]
\end{definition}

Strategies for player MIN are defined analogously and are usually denoted by $\tau$. We denote $\Sigma$ and $T$ the sets of strategies for players MAX and MIN respectively.

\begin{definition}[play]
    A \emph{play} is a sequence of vertices $x_0, x_1, x_2, \dots$ such that for all $t \geq 0$, 
    \[ (x_t, x_{t+1}) \in A.\]
    Such a play is \emph{consistent} with strategies $\sigma$ and $\tau$, respectively for player MAX and player MIN, if for all $t \geq 0$,
    \[ x_t \in V_{MAX} \Rightarrow x_{t+1} = \sigma(x_t) \]
    and
    \[ x_t \in V_{MIN} \Rightarrow x_{t+1} = \tau(x_t). \]
\end{definition}

A couple of strategies $\sigma, \tau$ and an initial vertex $x_0 \in V$ define recursively
a random play consistent with $\sigma, \tau$ by setting:

\begin{itemize}

    \item if $x_t \in V_{MAX}$ then $x_{t+1} = \sigma(x_t)$;
    \item if $x_t \in V_{MIN}$ then $x_{t+1} = \tau(x_t)$;
    \item if $x_t \in V_{SINK}$ then $x_{t+1} = x_t$;
    \item if $x_t \in V_{AVE}$, then $x_{t+1}$ is one of the two neighbours of $x_t$,
        the choice being made by a fair coin, independently of all other random choices. 
\end{itemize}

Hence, two strategies $\sigma, \tau$, together with an initial vertex $x_0$ define a measure of probability $\Proba_{\sigma,\tau}^{x_0}$ on plays consistent with $\sigma, \tau$. Note that if a play contains a sink vertex $x$, then at every subsequent time the play stays in $x$. Such a play is said to \emph{reach} sink $x$. 
To every play $x_0, x_1, \dots$ we associate a value which is the value of the sink reached by the play if any, and $0$ otherwise. This defines a random variable $X$ once two strategies are fixed. We are interested in the expected value of this quantity, which we call the value of a vertex $x \in V$ under strategies $\sigma, \tau$: 

\[ \val_{\sigma,\tau} (x) = {\mathbb E}_{\sigma,\tau}^x \left( X \right) \]
where ${\mathbb E}_{\sigma, \tau}^x$ is the expected value under probability $\Proba_{\sigma, \tau}^x$.
The goal of player $MAX$ is to maximize this (expected) value, and the best he can ensure against a strategy $\tau$ is

\[ \val_{\tau} (x) = \max_{\sigma \in \Sigma} \val_{\sigma,\tau} (x) \]
while against $\sigma$ player MIN can ensure that the expected value is at most
\[ \val_{\sigma} (x) = \min_{\tau \in T} \val_{\sigma,\tau} (x). \]

Finally, the value of a vertex $x$, is the common value
\begin{equation}
  \label{eq.value_function}
 \val(x) = \max_{\sigma \in \Sigma} \min_{\tau \in T} \val_{\sigma,\tau} (x) = \min_{\tau \in T} \max_{\sigma \in \Sigma} \val_{\sigma,\tau} (x).
\end{equation}
The fact that these two quantities are equal is nontrivial, and it can be found for instance in \cite{condon1992complexity}. A pair of strategies $\sigma^*, \tau^*$ such that, for all vertices $x$,
\[\val_{\sigma^*,\tau^*} (x) = \val(x)\]
always exists and these strategies are said to be \emph{optimal strategies}. It is polynomial-time equivalent to compute optimal strategies or to compute the values of all vertices in the game, since values can be obtained from strategies by solving a linear system. Conversely
if values are known, optimal strategies are given by greedy choices in linear time (see \cite{condon1992complexity} and Lemma \ref{lemma:precision}). Hence, we shall simply write "solve the game" for these tasks.

We shall need the following notion:
\begin{definition}[Stopping SSG]
    A SSG is said to be \emph{stopping} if for every couple of strategies all plays eventually reach a sink vertex with probability $1$.
\end{definition}

 Condon \cite{condon1992complexity} proved that every SSG $G$ can be reduced in polynomial time into a stopping SSG $G'$ whose size is quadratic in the size of $G$, and whose values almost remain the same. 

 \begin{theorem}[Optimality conditions, \cite{condon1992complexity}]\label{th:local_optimality_conditions}
    Let $G$ be a stopping SSG. The vector of values $(\val(x))_{x \in V}$ is the only vector $w$ 
    satisfying:
    \begin{itemize}
        \item for every $x \in V_{MAX}$, $w(x) = \max\{w(y) \mid (x,y) \in A\}$;
        \item for every $x \in V_{MIN}$, $w(x) = \min\{w(y) \mid  (x,y) \in A\} $;
        \item for every $x \in V_{AVE}$ $w(x) = \frac12 w(x^1) +\frac12 w(x^2)$ where $x^1$ and $x^2$ are the two neighbours of $x$; 
        \item for every $x \in V_{SINK}$, $w(x) = \val(x)$.
    \end{itemize}
\end{theorem}

If the underlying graph of an SSG is acyclic, then the game is stopping and the previous local optimality conditions yield a very simple way to compute values. Indeed, we can use backward propagation of values since all leaves are sinks, and the values of sinks are known. We naturally call these games \emph{acyclic SSGs}.\\

Once a pair of strategies has been fixed, the previous theorem enables us to see the values as solution of a system of linear equations. This yields the following lemma, which is an improvement on a similar result in \cite{condon1992complexity}, where the bound is $4^n$ instead of $6^{\frac{n_a}{2}}$.

\begin{lemma}\label{lemma:precision}
    Let $G$ be an SSG with sinks having rational values of common denominator $q$. 
    Then under any pair of strategies $\sigma,\tau$, the value $\val_{\sigma,\tau}(x)$ of any vertex $x$ can be computed in time $O(n_{a}^{\omega})$,
    where $\omega$ is the exponent of the matrix multiplication, and $n_a$
    the number of average (binary) vertices.
    Moreover, the value can  be written as a rational number $\frac{a}{b}$, with
    \[0 \leq a, b \leq 6^{\frac{n_a}{2}} \times q .\]
\end{lemma}

\begin{proof}
    We sketch the proof since it is standard.
    First, one can easily compute all vertices $x$ such that
    \[ \val_{\sigma, \tau} (x) = 0.\]
    Let $Z$ be the set of these vertices. Then: 
    \begin{itemize}
        \item all AVE vertices in $Z$ have all their neighbours in $Z$;
        \item all MAX (resp. MIN) vertices $x$ in $Z$ are such that $\sigma(x)$ (resp. $\tau(x)$) is in $Z$. 
    \end{itemize}

    To compute $Z$, we can start with the set $Z$ of all vertices except sinks with positive value and iterate the following
    \begin{itemize}
        \item if $Z$ contains an AVE vertex $x$ with a neighbour out of $Z$, remove $x$ from $Z$;
        \item if $Z$ contains a MAX (resp. MIN) vertex $x$ with $\sigma(x)$ (resp. $\tau(x)$) out of $Z$, remove $x$ from $Z$.
    \end{itemize} 
    This process will stabilize in at most $n$ steps and compute the required set $Z$. 
    Once this is done, we can replace all vertices of $Z$ by a sink with value zero, obtaining
    a game $G'$ where under $\sigma, \tau$, the values of all vertices will be unchanged.

    Consider now in $G$' two corresponding strategies $\sigma, \tau$ (keeping the same names to simplify) and a positional vertex $x$. 
    Let $x'$ be the first non positional vertex that can be reached from $x$ under strategies $\sigma, \tau$. Clearly, $x'$ is well defined and \[\val_{\sigma,\tau}(x) = \val_{\sigma, \tau}(x').\]
    This shows that the possible values under $\sigma, \tau$ of all vertices are the values of average and sink vertices.
    The same is true if one average vertex has its two arcs towards the same vertex, thus we can forget those also.    
    The value of an average vertex being equal to the average value of its two neighbours, we see that we can write a system
   
    \begin{equation} \label{z=az+b}
    z = Az +b 
    \end{equation}
    where
    \begin{itemize}
        \item $z$ is the $n_{a}$-dimensional vector containing the values of average vertices
        \item $A$ is a matrix where all lines have at most two $\frac12$ coefficients, the rest being zeros
        \item $b$ is a vector whose entries are of the form $0$, $\frac{p_i}{2q}$ or $\frac{p_i + p_j}{2q}$, corresponding to transitions from average vertices to sink vertices.
    \end{itemize}
    
    Since no vertices but sinks have value zero, it can be shown that this system has a unique solution, i.e. matrix $I-A$ is nonsingular, where $I$ is the $n_{a}$-dimensional identity matrix. We refer to \cite{condon1992complexity} for details, the idea being that since in $n-1$ steps there is a small probability of transition from any vertex of $G'$ to a sink vertex, the sum of all coefficients on a line of $A^{n-1}$ is strictly less than one, hence the convergence of
    \[ \sum_{k \geq 0} A^k = (I - A)^{-1}.\]

    Rewriting (\ref{z=az+b}) as  
\[ 2(I-A) z= 2b,\]
we can use Cramer's rule to obtain that the value $z_v$ of an average vertex $v$ is 
\[ z_v = \frac{\det B_v}{ \det 2(I-A)} \]
where $B_v$ is the matrix $2(I-A)$ with the column corresponding to $v$ replaced by $2b$. Hence
by expanding the determinant we see that $z_v$ is of the form
\[ \frac{1}{\det 2(I-A)} \sum_{w \in V_{AVE}} \pm 2b_w \det(2(I-A)_{v,w}) \] 
where $2(I-A)_{v,w}$ is the matrix $2(I-A)$ where the line corresponding to $v$ and the column
corresponding to $w$ have been removed.

Since $2b_w$ has either value $0$, $\frac{p_i}{q}$ or $\frac{p_i + p_j}{q}$
for some $1 \leq i,j \leq n_{a}$, we can write the value of $z_v$ as a fraction of integers
\[ \frac{ \sum_{w \in v} \pm 2b_w q \cdot \det(2(I-A)_{v,w})}{   \det 2(I-A) \cdot q}\] 
It remains to be seen, by Hadamard's inequality, that since the nonzero entries of $2(I-A)$
on a line are a $2$ and at most two $-1$, we have \[\det 2(I-A) \leq 6^{\frac{n_a}{2}} ,\] which concludes the proof.

\end{proof}

The bound $6^{\frac{n_a}{2}}$ is almost optimal. Indeed a caterpillar tree of $n$ average vertices connected to the $0$ sink 
except the last one, which is connected to the $1$ sink, has a value of $\frac{1}{2}^n$ at the root.
Note that the lemma is slightly more general (rational values on sinks) and the bound a bit better ($\sqrt{6}$ instead of $4$) than what is usually found in the literature.

In all this paper, the complexity of the algorithms will be given in term of number of arithmetic 
operations and comparisons on the values as it is customary. The numbers occurring in the algorithms are rationals of value at most exponential in the number of vertices in the game, therefore the bit complexity is increased by at most an almost linear factor.

\section{MAX-acyclic SSGs}\label{sec:max-acyclic}

  In this section we define a class of SSG that generalize acyclic SSGs and still
  have a polynomial-time algorithm solving the value problem.
  
  A cycle of an SSG is an \emph{oriented} cycle of the underlying graph.

  \begin{definition}
  We say that an SSG is \emph{MAX-acyclic} (respectively MIN-acyclic) if from any MAX vertex $x$ (resp. MIN vertex),
  for all outgoing arcs $a$ but one, all plays going through $a$ never reach $x$ again. 
  \end{definition}

  Therefore this class contains the class of acyclic SSGs and
  we can see this hypothesis as being a mild form of acyclicity. From now on, we will stick to 
  MAX-acyclic SSGs, but any result would be true for MIN-acyclic SSGs also.
  There is a simple characterization of MAX-acyclicity in term of the structure of the underlying graph.
  
  \begin{lemma}
   An SSG is  MAX-acyclic if and only if every MAX vertex has at most one 
  outgoing arc in a cycle.
  \end{lemma}

  Let us specify the following notion.

  \begin{definition}
      We say that an SSG is \emph{strongly connected} if the underlying directed graph, once sinks are removed,
      is strongly connected.
  \end{definition}

\begin{lemma}
    Let $G$ be a MAX-acyclic, strongly connected SSG. Then for each MAX $x$, all neighbours
    of $x$ but one must be sinks.
\end{lemma}

\begin{proof}
    Indeed, if $x$ has two neighbours $y$ and $z$ which are not sinks, then by strong connexity there are directed paths from $y$ to $x$ and from $z$ to $x$. Hence,
    both arcs $xy$ and $xz$ are on a cycle, contradicting the assumption of MAX-acyclicity.
\end{proof}

  From now on, we will focus on computing the values of a strongly connected MAX-acyclic SSG.
  Indeed, it easy to reduce the general case of a MAX-acyclic SSG to strongly connected by computing the DAG of the
  strongly connected components in linear time. We then only need to compute the values in each of the components, beginning by the leaves.

  %\subsection{Solving a strongly connected MAX-acyclic game}

  We will show that the Hoffman-Karp algorithm~\cite{hoffman1966nonterminating,condon1993algorithms}, when applied
  to a strongly connected MAX-acyclic SSG, runs for at most a linear number of steps before reaching an optimal solution.
  Let us remind the notion of \emph{switchability} in simple stochastic games. If $\sigma$ is a strategy for MAX, then a MAX vertex
  $x$ is \emph{switchable} for $\sigma$ if there is an neighbour $y$ of $x$ such that $\val_\sigma(y) > \val_\sigma( \sigma(x) )$.
  \emph{Switching} such a vertex $x$ consists in considering the strategy $\sigma'$, equal to $\sigma$ but for $\sigma'(x)=y$.

  For two vectors $v$ and $w$, we note $v \geq w$ if the inequality holds componentwise, and $v>w$ if moreover at least one component is strictly larger.

  \begin{lemma}[See Lemma 3.5 in~\cite{tripathi2011strategy}]
      Let $\sigma$ be a strategy for MAX and $S$ be a set of switchable vertices. Let $\sigma'$ be the strategy obtained when all
      vertices of $S$ are switched. Then
      \[ \val_{\sigma'} > \val_\sigma . \]
  \end{lemma}

   Let us recall the Hoffman-Karp algorithm:

   \begin{enumerate}
       \item Let $\sigma_0$ be any strategy for MAX and $\tau_0$ be a best response to $\sigma_0$
       \item while $(\sigma_t,\tau_t)$ is not optimal:
           \begin{enumerate}
               \item let $\sigma_{t+1}$ be obtained from $\sigma_t$ by switching one (or more) switchable vertex 
               \item let $\tau_{t+1}$ be a best response to $\sigma_{t+1}$
           \end{enumerate}
   \end{enumerate}

   The Hoffman-Karp algorithm computes a finite sequence $(\sigma_t)_{0 \leq t \leq T}$ of strategies for the MAX player such that
   \[        \forall 0 \leq t \leq T-1, \quad \val_{\sigma_{t+1}} > \val_{\sigma_t}. \]

If any MAX vertex $x$ in a strongly connected MAX-acyclic SSG has more than one sink neighbour, say $s_1, s_2, \cdots s_k$, then these can be replaced by a single sink neighbour $s'$ whose value is
\[ \val(s') := \max_{i=1..k} \val(s_i).\]
Hence, we can suppose that all MAX vertices in a strongly connected MAX-acyclic SSG have degree two.
For such a reduced game, we shall say that a MAX vertex $x$ is \emph{open} for a strategy $\sigma$ if $\sigma(x)$ is the sink
neighbour of $x$ and that $x$ is \emph{closed} otherwise.

   \begin{lemma}
     \label{lem.algo_max_acyclic}
       Let $G$ be a strongly connected, MAX-acyclic SSG, where all MAX vertices have degree 2. Then the Hoffman-Karp algorithm, starting from any strategy $\sigma_0$, halts in at most $2n_{M}$ steps. Moreover, starting from the strategy where every MAX vertex is open, the algorithm halts in at most $n_{M}$ steps. All in all, the computation is polynomial in the size of the game.
   \end{lemma}

   \begin{proof}
       We just observe that if a MAX vertex $x$ is closed at time $t$ ,
       then it remains so until the end of the computation. More precisely, if  $s := \sigma_{t-1}(x)$ is a sink vertex, and $y:=\sigma_{t}(x)$ is not,
       then since $x$ has been switched we must have
       \[\val_{\sigma_t}(y) > \val_{\sigma_t} (s). \]
       For all subsequent times $t'>t$, since strategies are improving we will have
       \[ \val_{\sigma_{t'}}(y) \geq \val_{\sigma_t}(y) > \val_{\sigma_t}(s) = \val_{\sigma_{t'}}(s) = \val(s),\]
       so that $x$ will never be switchable again.

       Thus starting from any strategy, if a MAX vertex is closed it cannot be opened and closed again, and if it is open 
       it can only be closed once.
   \end{proof}

   Each step of the Hoffman-Karp algorithm requires to compute a best-response for the MIN player.
   A best-response to any strategy can be simply computed with a linear program with as many variables as vertices in the SSG, hence in polynomial time. We will denote this complexity by $O(n^{\eta})$; it is well known that we can have $\eta \leq 4$, for instance with Karmarkar's algorithm. 

   \begin{theorem}
     \label{thm.complexity_max_acyclic}
     A strongly connected MAX-acyclic SSG can be solved in time $O(n_{M}n^{\eta})$.
   \end{theorem}

   Before ending this part, let us note that in the case where the game is also MIN-acyclic, one can compute directly a best response to a MAX strategy $\sigma$ without linear programming: starting with a MIN strategy $\tau_0$ where all MIN vertices are open, close all MIN vertices $x$ such that their neighbour has a value strictly less than their sink. One obtains a strategy $\tau_1$ such that
   \[ \val_{\sigma, \tau_1} < \val_{\sigma, \tau_0},\]
and the same process can be repeated. By a similar argument than in the previous proof, a closed MIN vertex will never be opened again, hence the number of steps is at most the number of MIN vertices, and each step only necessitates to compute the values, i.e. to solve a linear system (see Lemma~\ref{lemma:precision}).

\begin{corollary}
\label{cor.POS-acyclic}
     A strongly connected MAX and MIN-acyclic SSG can be solved in time $O(n_{M} n_{m} n^\omega)$, where $\omega$ is the exponent of matrix multiplication.
\end{corollary}

\section{SSG with few fork vertices}\label{sec:few_cycles}

Work on this section has begun with Yannis Juglaret during his Master internship at PRiSM laboratory. Preliminary results about SSGs with one simple cycle can be found in his report~\cite{yannisreport}.
We shall here obtain fixed-parameter tractable (FPT) algorithms in terms of parameters quantifying how far a graph is from being MAX-acyclic and MIN-acyclic, in the sense of section \ref{sec:max-acyclic}. These parameters are:

$$k_p = \displaystyle{\sum_{x \in V_{POS}} (|\{ y: (x,y)\in A \mbox{ and is in a cycle}\}| -1)} $$
and
$$k_a = \displaystyle{\sum_{x \in V_{AVE}} (|\{ y: (x,y) \in A \mbox{ and is in a cycle}\}| -1)} .$$

We say that an SSG is POS-acyclic (for \emph{positional} acyclic) when it is both MAX and MIN-acyclic. Clearly, parameter $k_p$ counts
the number of edges violating this condition in the game. Similarly, we say that the game is AVE-acyclic when average vertices have at most one outgoing arc in a cycle. 
We call \emph{fork vertices}, those vertices that have at least two outgoing arcs in a cycle. 
Since averages vertices have only two neighbours, $k_a$ is the number of fork average vertices.

Note that:
\begin{enumerate}
\item When $k_p = 0$ (respectively $k_a=0$), the game is POS-acyclic (resp. AVE-acyclic).
\item When $k_a = k_p =0$, the strongly connected components of the game are cycles. We study these games, which we call \emph{almost acyclic}, in detail in subsection \ref{section:almost}.
\item Finally, the number of simple cycles of the SSG is always less than $k_p + k_a$, therefore getting an FTP algorithm in $k_p$ and $k_a$ immediately gives an FTP algorithm in the number of cycles.
\end{enumerate}

We obtain:

\begin{theorem} \label{choucroute}
There is an algorithm which solves the value problem for SSGs in time $O(n f(k_p, k_a))$, with $f(k_p, k_a)= k_a!4^{k_a} 2^{k_p}$.
\end{theorem}

As a corollary, by remark 3 above we have:
\begin{theorem}\label{th:few-cycles}
There is an algorithm which solves the value problem for SSGs with $k$ simple cycles in time $O(n g(k))$ with $g(k)= (k-1)!4^{k-1}$.
\end{theorem}

Note that in both cases, when parameters are fixed, the dependance in $n$ is \emph{linear}.\\

Before going further, let us explain how one could easily build on the previous part and obtain an FPT algorithm in parameter $k_p$, but with a much worse dependance in $n$.

When $k_p>0$, one can fix partially a strategy on positional fork vertices, hence obtaining a POS-acyclic subgame that can be solved in polynomial time according to Corollary~\ref{cor.POS-acyclic}, using the Hoffman-Karp algorithm. 
Combining this with a bruteforce approach looking exhaustively through all possible local choices at positional fork vertices, we readily obtain a polynomial  algorithm for the value problem when $k_p$ is fixed:
\begin{theorem}
  \label{thm.cycles_HK}
  There is an algorithm to solve the value problem of an SSG in time $O(n_{M} n_{m} n^\omega 2^{k_p})$. 
\end{theorem}

We shall conserve this brute-force approach. 
In the following, we give an algorithm that reduces the polynomial complexity to a linear complexity when $k_a$ is fixed. From now on, up to applying the same bruteforce procedure, we assume $k_p = 0$ (all fork vertices are average vertices). We also consider the case of a strongly connected SSG, since otherwise the problem can be solved for each strongly connected component as done in Section~\ref{sec:max-acyclic}.  We begin with the baseline $k_a = 0$ and extend the algorithm to general values of $k_a$. Before this, we provide some preliminary lemmas and definitions that will be used in the rest of the section.

A \emph{partial strategy} is a strategy defined on a subset of vertices controlled by the player. Let $\sigma$ be such a partial strategy for player MAX, we denote by $G[\sigma]$ the subgame of $G$ where the set of strategies of MAX is reduced to the ones that coincide with $\sigma$ on its support. According to equation~(\ref{eq.value_function}), the value of an SSG is the highest expected value that MAX can guarantee, then it decreases with the set of actions of MAX:
\begin{lemma}
  \label{lem.subgame}
  Let $G$ be an SSG with value $v$, $\sigma$ a partial strategy of MAX, and $G[\sigma]$ the subgame induced by $\sigma$ with value $v'$. Then $v \geq v'$.
\end{lemma}

In strongly connected POS-acyclic games, positional vertices have at least one outgoing arc to a sink. Recall that, in case the strategy chooses a sink, we say that it is open at this vertex (the strategy is said open if it is open at a vertex), and closed otherwise. We can then compare the value of an SSG with that of the subgame generated by any open strategy. 
\begin{lemma}
  \label{lem.opencycle}
  Let $x$ be a MAX vertex of an SSG $G$ with a sink neighbour, and $\sigma$ the partial strategy open at $x$. If it is optimal to open $x$ in $G$, then it is optimal to open it in $G[\sigma]$. 
\end{lemma}

\begin{proof}
Since it is optimal to open $x$ in $G$, the value of its neighbour sink is at least that of any neighbour vertex, say $y$. But, in view of Lemma~\ref{lem.subgame}, the value of $y$ in $G[\sigma]$ is smaller than in $G$, and then it is again optimal to play a strategy open at $x$ in the subgame. 
\end{proof}

This lemma allows to reduce an almost acyclic SSG (resp. an SSG with parameter $k_a>0$) to an acyclic SSG (resp. an SSG with parameter $k_a-1$). Indeed, if the optimal MAX strategy is open at vertex $x$, then the optimal strategy of  the subgame open at any MAX vertex will be open at $x$. A solution to find $x$ once the subgame is solved (and then to reduce the parameter $k_a$) consists in testing all the open MAX vertices. But it may be the case that all MAX vertices are open which would not yield a FPT algorithm. In Lemma~\ref{lem.1cycle} (resp. Lemma~\ref{lem.cycles}), we give a restriction on the set of MAX vertices that has to be tested  when $k_a=0$ (resp. $k_a>0$) which provides an FPT algorithm.

\subsection{Almost acyclic SSGs} \label{section:almost}

We consider an SSG with $k_a = 0$. Together with the hypothesis that it is POS-acyclic and strongly connected, its graph, once sinks are removed, consists of a single cycle. A naive algorithm to compute the value of such SSG consists in looking for, if it exists, a vertex that is open in the optimal strategy, and then solve the acyclic subgame:
\begin{enumerate}
\item For each positional vertex $x$:
  \begin{enumerate}
  \item compute the values of the acyclic SSG $G[\sigma]$, where $\sigma$ is the partial strategy open at $x$,
  \item if the local optimality condition is satisfied for $x$ in $G$, return the values.
  \end{enumerate}
\item If optimal strategies have not been found, return the value when all vertices are closed.
\end{enumerate}
This naive algorithm uses the routine that computes the value of an SSG with only one cycle. When the strategies are closed, the values can be computed in linear time as for an acyclic game. Indeed, let $x$ be an average vertex (if none, the game can be solved in linear time) and  $s_1 \dots s_\ell$ be the values of the average neighbour sinks  in the order given by a walk on the cycle starting from $x$. Then the value of $x$ satisfies the equation $\val(x) = \frac12 s_1 + \frac12 ( \frac12 s_2 + \frac12 ( \dots +\frac12 (\frac12s_\ell + \frac12 \val(x)))   )$, so that  
\begin{equation}
\label{eq.value_cycle}
\val(x) = \frac{2^\ell}{2^\ell -1} \sum_{i=1}^\ell 2^{-i} s_i,
\end{equation}
which can be computed in time linear in the size of the cycle. The value of the other vertices can be computed by walking backward from $x$, again in linear time. Finally, since solving an acyclic SSG is linear, the complexity of the algorithm is $O(n^2)$  which is still better than the complexity $O(n_{M}n_{m}n^\omega)$ obtained with the Hoffman-Karp algorithm (see Theorem~\ref{thm.cycles_HK}). 
% We denote by \emph{ComputeValues1Cycle} the routine that compute the value in this case. Finally we denote by \emph{testOptimality(v,x,G)} the boolean function that is true iff the local optimality of a vector of values at vertex $x$ in game $G$ is satisfied (if $x$ is omitted in the function, then it refers to the global optimality conditions).

% \begin{algorithm}[t]
% \For
% {$x \in V_{POS}$}
% {
%   $\sigma$ = strategy open at $x$
%   \\[2pt]
%   v = ComputeValuesAcyclic($G[\sigma]$)
%   \\[2pt]
%   \If{ testOptimality($v,x,G$) }{
%     \Return{$v$}
%   }%If
% }%ForEach
% \Return ComputeValues1Cycle($G$)
% \caption{Naive algorithm in $O(n^2)$ for computing the values of an almost acyclic SSG $G$}
% \label{algo.naive_1_cycle}
% \end{algorithm}

Remark that this algorithm can readily be extended to a SSG with $k$ cycles with a complexity $O(n^{k+1})$. Hence it is not an FPT algorithm for the number of cycles.
However, we can improve on this naive algorithm by noting that the optimal strategy belongs to one of the following subclasses of strategies:
\begin{itemize}
\item[$(i)$] strategies closed everywhere,
\item[$(ii)$] strategies open at least at one MAX vertex,
\item[$(iii)$] strategies open at least at one MIN vertex.
\end{itemize}
The trick of the algorithm is that, knowing which of the three classes the optimal strategy belongs to, the game can be solved in linear time. Indeed:
\begin{itemize}
\item[$(i)$] If the optimal strategy is closed at every vertex, the value can be computed in linear time as shown before.
\item[$(ii)$] If the optimal strategy is open at a MAX vertex (the MIN case is similar), then it suffices to solve in linear time the acyclic game $G[\sigma]$ where $\sigma$ is any partial strategy open at a MAX vertex, and then use the following Lemma to find an open vertex for the optimal strategy of the initial game.
\end{itemize}

\begin{lemma}
  \label{lem.1cycle}
  Let $G$ be a strongly connected almost-acyclic SSG.  Assume that the optimal strategy is open at a MAX vertex. For any partial strategy $\sigma$ open at a MAX vertex $x$, let $x = x_0, x_1 \dots x_\ell = x$ be the sequence of the $\ell$ open MAX vertices for the optimal strategy of $G[\sigma]$ listed in the cycle order. Then it is optimal to open $x_{1}$.
\end{lemma}

\begin{proof}
  Let $\bar{x}$ be a MAX vertex that is open when solving $G$. From Lemma~\ref{lem.opencycle}, there is an index $i$ such that $x_i = \bar{x}$ (in particular there exists an open MAX vertex  when solving $G[\sigma]$). If $\ell = 1$ then $x_0 = \bar{x}$ so that the optimal strategies of $G[\sigma]$ and $G$ coincide. Otherwise, if $i = 1$, the result is immediate. At last, if $i>1$, $x_i$ has the same value in $G$ and $G[\sigma]$ (the value of its sink), and so has the vertex just before if it is different from $x_0$. Going backward from $x_i$ in the cycle, all the vertices until $x_0$ (not included) have the same value in $G$ and $G[\sigma]$. In particular, if $i>1$, this is the case for $x_1$ whose value is then the value of its sink. So it is optimal to open $x_1$ in $G$ as well.
\end{proof}

% The linear algorithm that solves a strongly connected almost-acyclic SSG $G$ is given at Algorithm~\ref{algo.1_cycle}. It uses the function \emph{ComputeOptStrat(player, G)} that computes the optimal strategy of the specified player in game $G$.

All in all, a linear algorithm that solves a strongly connected almost-acyclic SSG $G$ is:
\begin{enumerate}
\item Compute the values of the strategies closed everywhere. If optimal, return the values.
\item Else compute the optimal strategies of $G[\sigma_1]$ where $\sigma_1$ is a partial strategy open at a MAX vertex $x$; let $y$ be the first open MAX vertex after $x$; compute the values of $G[\sigma_2]$ where  $\sigma_2$ is the partial strategy open at $y$; if the local optimality condition is satisfied for $y$ in $G$, return the values.
\item Else apply the same procedure to any MIN vertex.
\end{enumerate}

% \begin{algorithm}[t]
%   \tcc{Test optimality of closed strategies}
%   $v$ = ComputeValues1Cycle($G$)
%   \\[2pt]
%   \If{ testOptimality($v,G$)}{
%     \Return{$v$}
%   }%If
%   \tcc{Test if  optimal to open one MAX vertex}
%   $\sigma$ = strategy open at MAX vertex $x$
%   \\[2pt]
%   $\sigma'$ = ComputeOptStrat(MAX, $G[\sigma]$)
%   \\[2pt]
%   $y$ = first open MAX vertex in $\sigma'$ after $x$
%   \\[2pt]
%   $\sigma$ = strategy open at MAX vertex $y$
%   \\[2pt]
%   v = ComputeValuesAcyclic($G[\sigma]$)
%   \\[2pt]
%   \If{ testOptimality($v,x,G$) }{
%     \Return{$v$}
%   }%If
%   \tcc{Last case: all MAX vertices are closed and at least a MIN vertex is open}
%   $\tau$ = strategy open at MIN vertex $x$
%   \\[2pt]
%   $\tau'$ = ComputeOptStrat(MIN, $G[\tau]$)
%   \\[2pt]
%   $y$ = first open MIN vertex in $\tau'$ after $x$
%   \\[2pt]
%   $\tau$ = strategy open at MIN vertex $y$
%   \\[2pt]
%   v = ComputeValuesAcyclic($G[\tau]$)
%   \\[2pt]
%   \If{ testOptimality($v,x,G$) }{
%     \Return{$v$}
%   }%If
% \caption{Algorithm in $O(n)$ for computing the values of an almost acyclic SSG $G$}
% \label{algo.1_cycle}
% \end{algorithm}

\begin{theorem}
  \label{thm.1cycle}
  There is an algorithm to solve the value problem of a strongly connected almost-acyclic SSG in time $O(n)$ with $n$ the number of vertices. 
\end{theorem}

\subsection{Fixed number of non acyclic average vertices}

Again, we assume that the SSG is strongly connected and POS-acyclic. The algorithm for almost-acyclic games can be generalized as follow.

Firstly, it is possible to compute the values of the strategies closed everywhere in polynomial time in $k_a$ and check if this strategy is optimal. Indeed the value of each fork vertex can be expressed as an affine function of the value of all fork vertices in the spirit of Eq.~(\ref{eq.value_cycle}). Then the linear system of size $k_a$ can be solved in polynomial time, and the value of the remaining vertices computed by going backward from each fork vertex. This shows that computing the values of a game once strategies are fixed is polynomial in the number of fork average vertices, which slightly improves the complexity of Lemma~\ref{lemma:precision}.

Otherwise, the following Lemma allows to find a positional vertex that is open for the optimal strategy. We say that a vertex $x$ is the last (resp. first)  vertex in a set $S$ before (resp. after) another vertex $y$ if there is a unique simple path from $x$ to $y$ (resp. $y$ to $x$) that does not contain any vertex in $S$, $x$ and $y$ being excluded.
\begin{lemma}
  \label{lem.cycles}
  Let $G$ be a strongly connected SSG with a set $A=\{a_1, \dots, a_\ell \}$ of fork average vertices and no positional fork vertices. Assume that the optimal strategy of $G$ is open at a MAX vertex. Let $\sigma$ be a partial strategy open at any vertex that is the last MAX vertex before a vertex in $A$. Let $S[\sigma]$ be the set of open MAX vertices when solving  $G[\sigma]$. Then, there exists $x \in S[\sigma]$ that satisfies
  \begin{itemize}
  \item it is optimal to open $x$ in $G$,
  \item $x$ is the first vertex in $S[\sigma]$ after a vertex in $A$.
  \end{itemize}
\end{lemma}

\begin{proof}
  Let $\bar{x}$ be a MAX vertex that is open when solving $G$, and  $i$ be such that $a_i$ is the last vertex in $A$ before $\bar{x}$. By Lemma~\ref{lem.opencycle}, $\bar{x}$ is open in $G[\sigma]$ and since $\sigma$ is not open at a MAX vertex between $a_i$ and $\bar{x}$, all the vertices between the successor of $a_i$  and $\bar{x}$ have the same value in $G[\sigma]$ and $G$, and then the optimal strategy of $G[\sigma]$ at these vertices is optimal for $G$. This property holds in particular for the first vertex in $S[\sigma]$ after $a_i$ in the path leading to $\bar{x}$.
\end{proof}
The Lemma is illustrated on Figure~\ref{fig:lemma}.

\begin{figure}
  \begin{center}
    \begin{tikzpicture}[scale=0.8]      
      \tikzstyle{block} = [rectangle, draw, fill=blue!20, text width=2em, text centered, rounded corners, minimum height=2em]
      \tikzstyle{aver} = [circle, draw, fill=red!20, text centered, minimum height=2em]
      \tikzstyle{final} = [rectangle, draw, fill=yellow!30, text width=2em, text centered, rounded corners, minimum height=2em]

      \node[block] (M1) at (0,0) {$\mbox{max}_1$} ;
      \node[block] (a2) at (2,0) {$\dots$} ;
      \node[block] (M2) at (4,0) {$\mbox{max}_k$} ;
      \node[final] (1) at (0,2) {$s_1$} ;
      \node[final] (2) at (2,2) {$\dots$} ;
      \node[final] (3) at (4,2) {$s_k$} ;
      \node[aver] (a3) at (0,-2) {$a_1$} ;
      \node[aver] (a5) at (4,-2) {$a_2$} ;
      \node[aver] (a4) at (2,-2) {$\dots$} ;

      \draw[->,line width =1pt] (a5) edge (a4) ; 
      \draw[->,line width =1pt] (a4) edge (a3) ; 
      \draw[->,line width =1pt] (a3) edge (M1) ; 
      \draw[->,line width =1pt] (M1) edge (1) ; 
      \draw[->,line width =1pt] (M1) edge (a2) ; 
      \draw[->,line width =1pt] (M2) edge (3) ; 
      \draw[->,line width =1pt] (M2) edge (a5) ; 
      \draw[->,line width =1pt] (a2) edge (M2) ; 
      \draw[->,line width =1pt] (a2) edge (2) ; 
      \draw[loop below,line width =1pt, dotted] (a3) edge  (a3);
      \draw[loop below,line width =1pt, dotted] (a5) edge  (a5);
      \draw[loop below,line width =1pt, dotted] (a4) edge  (a4);
\end{tikzpicture}
  \end{center}
  \caption{Illustration of Lemma~\ref{lem.cycles}: $a_1$ and $a_2$ are average fork vertices. Vertices on the top are MAX vertices, labelled from 1 to $k$, that lead to sinks. $\max_k$ is the last MAX vertex before a fork vertex. Assume it is optimal to open at least one MAX vertex. Then, the first open MAX vertex (wrt to the labelling) of the optimal strategy of the subgame open at $\max_k$ is open as well in the optimal strategy of the initial game.}
  \label{fig:lemma}
\end{figure}
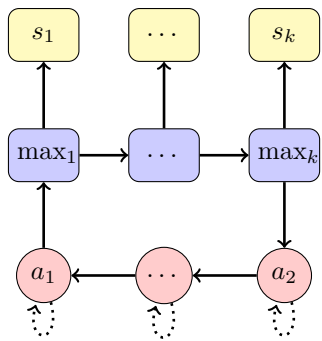

Finally, if the optimal strategy is open at some MAX vertex, then the following algorithm can be run to compute the values of $G$:
\begin{enumerate}
\item Let $x$ be the last MAX vertex before some fork vertex, and $\sigma_1$ the partial strategy open at $x$.  $G[\sigma_1]$ is an SSG that has $k_a - 1$ fork vertices (recall that $G$ is strongly connected). When solved, it provides a set $S[\sigma_1]$ of open MAX vertices. There are at most $k_a + k_a -1$ vertices that are the first in $S[\sigma_1]$ after a fork vertex. Then, from Lemma~\ref{lem.cycles}, it is optimal to open at least one of them in $G$. 
\item For each $y $ that is the first in $S[\sigma_1]$ after a fork vertex:
  \begin{enumerate}
  \item compute the values of $G[\sigma_2]$, $\sigma_2$ being the partial strategy open at $y$,
  \item if local optimality condition is satisfied for $y$ in $G$, return the values.
  \end{enumerate}
\end{enumerate}

This algorithm computes at most $2k_a$ SSGs with $k_a-1$ average fork vertices. In the worst case, the same algorithm must be run for the MIN vertices. Using theorem~\ref{thm.1cycle} for the case $k_p=k_a=0$, we obtain Theorem \ref{choucroute} and its corollary.

\section{Feedback vertex set}\label{sec:feedback}

A feedback vertex set is a set of vertices in a directed graph such that removing them
yields a DAG. Computing a minimal vertex set is an $\NP$-hard problem~\cite{gary1979computers}, but it 
can be solved with a FPT algorithm~\cite{chen2008fixed}. %note the factor is 4^k k! as in our previous algorithm
Assume the size of the minimal vertex set is fixed, we prove in this section that we can find 
the optimal strategies in polynomial time.
Remark that, to prove such a theorem, we cannot use the result on bounded number of cycles since a DAG plus one vertex may have an exponential number of cycles. 
Moreover a DAG plus one vertex may have a large number of positional vertices with several arcs in a cycle, thus we cannot use the algorithm to solve MAX-acyclic
plus a few non acyclic MAX vertices.

The method we present works by transforming $k$ vertices into sinks
and could thus be used for other classes of SSGs. For instance, it could solve in polynomial time the value problem for games which are MAX-acyclic when $k$ vertices are removed.

\subsection{The dichotomy method}

We assume from now on that all SSGs are stopping.
In this subsection, we explain how to solve an SSG by solving 
it several times but with one vertex less.

First we remark that turning any vertex into a sink of its own value in the original game
does not change any value. 

\begin{lemma}\label{lemma:sink}
Let $G$ be an SSG and $x$ one of its vertex. 
Let $G'$ be the same SSG as $G$ except that $x$ has been turned into a sink
vertex of value $\val_G(x)$. For all vertices $y$, $\val_G(y) = \val_{G'}(y)$.
\end{lemma}

\begin{proof}
 The optimality condition of Theorem~\ref{th:local_optimality_conditions} are exactly the same in 
 $G$ and $G'$. Since the game is stopping, there is one and only one solution to these equations and
 thus the values of the vertices are identical in both games.
\end{proof}

The values in an SSG are monotone with regards to the values of the sinks, as proved in the next lemma.
\begin{lemma}\label{lemma:increasing}
 Let $G$ be an SSG and $s$ one of its sink vertex. 
 Let $G'$ be the same SSG as $G$ except that the value of $s$ has been increased. 
 For all vertices $x$, $\val_G(x) \leq \val_{G'}(x)$.
\end{lemma}
\begin{proof}
 Let fix a pair of strategy $(\sigma,\tau)$ and a vertex $x$.
 We have: $$\val_{(\sigma,\tau),G}(x) = \sum_{y \in V_{SINK}}\pro{x \rightsquigarrow y} \val_G(y)$$
         $$\val_{(\sigma,\tau),G}(x) \leq  \sum_{y \in V_{SINK}} \pro{x \rightsquigarrow y} \val_{G'}(y) = \val_{(\sigma,\tau),G'}(x)$$
 because $\val_G(x) = \val_{G'}(x)$  except when $x = s$, $\val_G(s) \leq \val_{G'}(s)$. Since the inequality is true for every pair of strategies and every vertex, the lemma is proved. 
\end{proof}

Let $x$ be an arbitrary vertex of $G$ and  let $G[v]$  be the same SSG, except that $x$ becomes a SINK vertex of value $v$.
We define the function $f$ by: 

$$\left\{ \begin{array}{l l}
        \mbox{if x is a MAX vertex,}  & f(v)= \max\{\val_{G[v]}(y) : (x,y) \in A\} \\
        \mbox{if x is a MIN vertex,}  & f(v)=\min\{\val_{G[v]}(y): (x,y) \in A\}   \\
        \mbox{if x is an AVE vertex,} & f(v)=\frac{1}{2}\val_{G[v]}(x^1) + \val_{G[v]}(x^2) \\
       \end{array} \right.$$

\begin{lemma}\label{lemma:dichotomy}
There is a unique $v_0$ such that $f(v_0) = v_0$ which is $v_0 = \val_G(x)$.
Moreover, for all $v > v_0$, $f(v_0) < v_0$ 
and for all $v < v_0$, $f(v_0) > v_0$.
\end{lemma}

\begin{proof} 
The local optimality conditions given in Theorem~\ref{th:local_optimality_conditions} are the same in $G$ and $G[v]$ except the equation $f(\val_G(x)) = \val_G(x)$. Therefore, when $f(v_0) = v_0$, the values of $G[v]$ satisfy all the local optimality conditions of $G$. Thus $v_0$ is the value of $s$ in $G$. Since the game is stopping there is at most one such value.
 
 Conversely, let $v_0$ be the value of $s$ in $G$. 
 By Lemma \ref{lemma:sink}, the values in $G[v_0]$ are the same as in $G$ for all vertices. Therefore the local optimality conditions 
 in $G$ contains the equation $f(v_0) = v_0$.  
 
 We have seen that $f(v) = v$ is true for exactly one value of $v$. 
 Since the function $f$ is increasing by Lemma \ref{lemma:increasing} and because $f(0) \geq 0$ and $f(1) \leq 1$,
 we have for all $v > v_0$, $f(v_0) < v_0$ and for all $v < v_0$, $f(v_0) > v_0$.
\end{proof}

The previous lemma allows to determine the value of $x$ in $G$
by a dichotomic search by the following algorithm. 
We keep refining an interval $[min,max]$ which contains the value of $x$, 
with starting values $min =0$ and  $max = 1$.

\begin{enumerate}
 \item While $max - min  \leq 6^{-n_a}$ do:
   \begin{enumerate}
   \item $v = (min + max)/2$
   \item Compute the values of $G[v]$
   \item If $f(v) > v$ then $min = v$
   \item If $f(v) < v $ then $max =v$
   \end{enumerate}
 \item Return the unique rational in $[min, max]$ of denominator less than $6^{-\frac{n_a}{2}}$	 
\end{enumerate}

\begin{theorem}\label{th:dichotomy_complexity}
 Let $G$ be an SSG with $n$ vertices and $x$ one of its vertex. Denote by $C(n)$ the complexity to solve 
 $G[v]$, then we can compute the values of $G$ in time $O(nC(n))$. In particular an SSG which can be turned into a 
 DAG by removing one vertex can be solved in time $O(n^2)$.
\end{theorem}
\begin{proof}
Let $v_0$ be the value of $x$ in $G$, which exists since the game is stopping.
 By Lemma~\ref{lemma:dichotomy} it is clear that the previous algorithm is such that $v_0$ is in the interval $[min,max]$ at any time. Moreover, by Lemma~\ref{lemma:precision} we know that $v_0 = \frac{a}{b}$ where $b \leq 6^{\frac{n_a}{2}}$.
 At the end of the algorithm, $max - min  \leq 6^{-n_a}$ therefore there is at most one rational of denominator less than
 $6^{\frac{n_a}{2}}$ in this interval. It can be found exactly with $O(n_a)$ arithmetic operations by doing a binary search in the Stern-Brocot tree (see for instance~\cite{graham1989patashnik}).
 
 One last call to $G[v_0]$ gives us all the exact values of $G$. 
 Since the algorithm stops when $max - min  \leq 6^{-n_a}$, we have at most $ O(n_a)$ calls
 to the algorithm solving $G[v]$. All in all the complexity is $O(n_aC(n) + n_a)$ that is $O(n_aC(n))$.

In the case where $G[v]$ is an acyclic graph, we can solve it in linear time which gives us the
stated complexity.
\end{proof}

\subsection{Feedback Vertex Set of Fixed Size }

Let $G$ be an SSG such that $X$  is one of its minimal vertex feedback set.
Let \mbox{$k = |X|$}. The game is assumed to be stopping. Since the classical 
transformation~\cite{condon1992complexity} into a stopping game does not change the size of a minimal vertex feedback set,
it will not change the polynomiality of the described algorithm. However the transformation produces an SSG which is quadratically larger, thus a good way to improve the algorithm we present would be to relax the stopping assumption.

In this subsection we will consider games whose sinks have dyadic values,
since they come from the dichotomy of the last subsection. The gcd
of the values of the sinks will thus be the maximum of the denominators.
The idea to solve $G$ is to get rid of $X$, one vertex at a time by the previous technique.
The only thing we have to be careful about is the precision up to which we have to do the 
dichotomy, since each step adds a new sink whose value has a larger denominator.

\begin{theorem} \label{th:feedback}
There is an algorithm which solves any stopping SSG in time $O(n^{k+1})$
where $n$ is the number of vertices and $k$ the size of the minimal feedback vertex set. 
\end{theorem}
\begin{proof}
 First recall that we can find a minimal vertex  with an FPT algorithm. You can also check every set of size $k$ 
 and test in linear time whether it is a feedback vertex set. Thus the complexity of finding such a set, that we denote by $X = \{x_1,\dots,x_k\}$, is at worst 
 $O(n^{k+1})$.
 Let denote by $G_i$ the game $G$ where $x_1$ to $x_i$ has been turned 
 into sinks of some values. If we want to make these values explicit we write  $G_i[v_1,\dots,v_i]$ where $v_1$ 
 to $v_i$ are the values of the sinks.
 
 We now use the algorithm of Theorem~\ref{th:dichotomy_complexity} recursively, that is we apply it 
 to reduce the problem of solving $G_i[v_1,\dots,v_i]$ to  the problem of solving $G_{i+1}[v_1,\dots,\\v_i,v_{i+1}]$ for several 
 values of $v_{i+1}$. Since $G_k$ is acyclic, it can be solved in linear time. Therefore the only thing we have to evaluate
 is the number of calls to this last step. To do that we have to explain how precise should be the dichotomy to solve $G_i$,
 which will give us the number of calls to solve $G_i$ in function of the number of calls to solve $G_{i+1}$.

 We prove by induction on $i$ that the algorithm, to solve $G_{i}$, makes $\log(p_i) $ calls to solve $G_{i+1}$,
 where the value $v_{i+1}$ is a dyadic number of numerator bounded by $p_{i} = 6^{(2^{i+1}-1)n_a}$.
 Theorem~\ref{th:dichotomy_complexity} proves the case $i = 0$. 
 Assume the property is proved for $i-1$, we prove it for $i$. By induction hypothesis, all the
 denominators of $v_1, \dots, v_{i}$ are power of two and their gcd is bounded by $p_i$.
 By Lemma~\ref{lemma:precision}, the value of $x_{i}$ is a rational of the form $\frac{a}{b}$ where $b \leq p_i6^{\frac{n_a}{2}}$.
 We have to do the dichotomy up to the square of $p_i6^{\frac{n_a}{2}}$ to recover the exact value of $x_{i}$ in the game $G_i(v_1,\dots,v_{i-1})$.
 Thus the bound on the denominator of $v_{i+1}$ is $p_{i+1} = p_i^2 6^{n_a}$.
 That is $p_{i+1} = 6^{2(2^{i+1}-1)n_a}6^{n_a} = 6^{(2^{i+2} -1)n_a} $, which proves the induction hypothesis.
 Since we do a dichotomy up to a precision $p_{i+1}$, the number of calls is clearly $\log(p_{i+1})$. 
  
 In conclusion, the number of calls to $G_k$ is  $$\displaystyle{\prod_{i=0}^{k-1} \log(6^{(2^{i+1} -1)n_a}) \leq 2^{k^2}\log(6)^{k}n_{a}^k}.$$
 Since solving a game $G_k$ can be done in linear time the total complexity is in $O(n^{k+1})$.
\end{proof}

\emph{Acknowledgements}  
This research was supported by grant ANR 12 MONU-0019 (project MARMOTE). Thanks to Yannis Juglaret for being so motivated to learn about SSGs and to Luca de Feo for insights on rationals and their representations.

\bibliographystyle{plain} 
\bibliography{acyclic_ssg.bib}

\end{document}